\documentclass{beatcs}

\newtheorem{proposition}{Proposition}
\newtheorem{lemma}[proposition]{Lemma}
\newtheorem{example}[proposition]{Example}
\newtheorem{sta}[proposition]{Statement}

\title{A note on syndeticity, recognizable sets and Cobham's theorem}
\author{Michel Rigo, Laurent Waxweiler\thanks{University of Li\`ege,
    Department of Mathematics, Grande Traverse 12 (B 37), B-4000
    Li\`ege, Belgium. \texttt{M.Rigo@ulg.ac.be}}}

\date{}

\begin{document}
\maketitle

\begin{abstract}
    In this note, we give an alternative proof of the following
    result.  Let $p,q\ge 2$ be two multiplicatively independent
    integers. If an infinite set of integers is both $p$- and
    $q$-recognizable, then it is syndetic. Notice that this result is
    needed in the classical proof of the celebrated Cobham's theorem.
    Therefore the aim of this paper is to complete \cite{Pe} and
    \cite{AS} to obtain an accessible proof of Cobham's theorem.
\end{abstract}

\section{Introduction}
Cobham's theorem is related to numeration systems and can be
considered as a classical result in formal languages theory. It is
formulated as follows. Let $p,q\ge 2$ be two multiplicatively
independent integers (i.e., the only integers satisfying $p^k=q^\ell$
are $k=\ell=0$).  If a subset $X\subseteq\mathbb{N}$ of integers is
both $p$- and $q$-recognizable then it is a finite union of arithmetic
progressions (i.e., $X$ is an {\it ultimately periodic} set). Recall
that $X\subset\mathbb{N}$ is said to be {\it $p$-recognizable} if the
language $\rho_p(X)$ of the $p$-ary representations (without leading
zeroes) of the elements in $X$ is a regular language accepted by a
finite automaton (see for instance \cite[Chap. 5]{Ei}).  This famous
result has been widely studied from various points of view (we give
here just a few references): extension to non-standard numeration
systems \cite{Du,Ha} or to the framework of $k$-regular sequences
\cite{bell}, study of the multidimensional case (known as
Cobham-Semenov's theorem) \cite{BHMV,PB}, alternative proofs using the
formalism of the first order logic \cite{Bes,MV}, \ldots.

The original proof due to Cobham is widely considered as rather
difficult \cite{Co}. In his book, S.~Eilenberg proposed as a challenge
to find an easier proof \cite{Ei}. The major improvements in the
simplification of the proof of Cobham's theorem were made by G. Hansel
in \cite{Ha1} where he makes use of the notion of syndeticity and
sketches the key-points leading to the result. Recall that an infinite
set of integers $X=\{x_0<x_1<\cdots\}$ is said to be {\it syndetic} if
there exists $C>0$ such that for all $n\ge 1$, $x_n-x_{n-1}\le C$.
(Notice that Hansel's ideas about syndeticity also hold in a wider
framework than $p$-ary numeration systems \cite{Ha2}.)

Afterwards, a great work of presentation relying on the main ideas
found in \cite{Ha1} was made by several authors \cite{AS,Pe}.
Unfortunately, in these last two documents a same mistake can be found
(Statement \ref{lem:faux} below is not correct and
Example \ref{exa:c} is a counter-example). In this note, our modest
contribution is to correct this error using as simple arguments as
possible. In the spirit, we are naturally close to \cite{Co} and
\cite{Ha1} but new ideas appear in our reasoning.  Finally, we hope
that this erratum added to \cite{Pe} or \cite{AS} will now give a
complete presentation of the proof of Cobham's theorem.

Let us set $\Sigma_p:=\{0,\ldots,p-1\}$ as the alphabet of the $p$-ary
digits. In \cite{AS,Pe}, the following result is presented.
\begin{sta}\label{lem:faux}
  If an infinite $p$-recognizable set $X\subseteq\mathbb{N}$ is such
  that $0^*\rho_p(X)$ is right dense, i.e., for all $u\in\Sigma_p^*$
  there exists $v\in\Sigma_p^*$ such that $uv\in0^*\rho_p(X)$, then
  $X$ is syndetic.
\end{sta}

\begin{example}\label{exa:c}
  As stated above, Statement \ref{lem:faux} is not correct.  An easy
  counter-example is given by the following set $X$ of integers
  $$X=\bigcup_{i\ge 0}[2^{2i},2^{2i+1}[.$$
  Indeed, this set is
  $2$-recognizable : $\rho_2(X)=1\{00,01,10,11\}^*$, and trivially
  right dense but not syndetic.
\end{example}
In the literature, Statement \ref{lem:faux} is generally presented to
obtain the following proposition.

\begin{proposition}\label{pro:1}\cite[Prop. 5]{Ha1}
  Let $p,q\ge 2$ be two multiplicatively independent integers. If an
  infinite set of integers if both $p$- and $q$-recognizable, then it
  is syndetic.
\end{proposition}
In substance, this latter result can naturally be found in Cobham's
work (see \cite[Lemma 3]{Co}). In this note, our aim is to give an
alternative proof of Proposition~\ref{pro:1} not using
Statement~\ref{lem:faux}. Our approach relies on five easy lemmas.

\section{Proof of the result}

We assume that the reader has some basic knowledge in automata theory
(see for instance \cite{Ei}). If $X\subseteq\mathbb{N}$ is a set of
integers, we define a mapping (or a right-infinite word)
$\mathbf{1}_X:\mathbb{N}\to\{0,1\}$ such that $\mathbf{1}_X(n)=1$ if
and only if $n\in X$. If $w$ is a finite word, $|w|$ denotes its length.

This first lemma will be useful in the proof of Lemma~\ref{lem:3} and
\ref{lem:4}.
\begin{lemma}\label{lem:1}
  Let $\mathcal{A}=(Q,q_0,F,\Sigma,\delta)$ be a DFA (Deterministic
  Finite Automaton) with $\delta:Q\times\Sigma^*\to Q$ as transition
  function. For any state $s\in Q$, the set
  $$L_s:=\{|w|\in\mathbb{N} : w\in\Sigma^*, \delta(s,w)\in F\}$$
  is
  such that $\mathbf{1}_{L_s}$ is ultimately periodic, i.e., there
  exist $N\ge 0$ and $P>0$ such that for all $n\ge N$,
  $\mathbf{1}_{L_s}(n)=\mathbf{1}_{L_s}(n+P)$.
\end{lemma}

\begin{proof}
  For any state $s\in Q$, we define a mapping $$f_s:\mathbb{N}\to
  \mathcal{P}(Q):n\mapsto\{\delta(s,w):w\in\Sigma^n\}.$$
  Since
  $\mathcal{P}(Q)$ is finite, there exist $a_s$ and $b_s$ such that
  $a_s<b_s$ and $f_s(a_s)=f_s(b_s)$. Obviously, for any
  $u,v\in\Sigma^*$, $\delta(s,uv)=\delta(\delta(s,u),v)$. Consequently
  for all $n\ge 0$,
  $$f_s(a_s+n)=\bigcup_{r\in f_s(a_s)}f_r(n)=\bigcup_{r\in
    f_s(b_s)}f_r(n)=f_s(b_s+n).$$
  In other words, $f_s$ is ultimately
  periodic: $f_s(n)=f_s(n+b_s-a_s)$ if $n\ge a_s$. To conclude the
  proof, observe that $\mathbf{1}_{L_s}=\mathbf{1}_{F_s}$ where
  $F_s=\{n\in\mathbb{N}:f_s(n)\cap F\neq\emptyset\}$.
\end{proof}

\begin{lemma}\label{lem:2}
    Let $m,n,a,b,c,d\in\mathbb{N}\setminus\{0\}$ be arbitrary integers
    such that $n<m$ and $p,q$ be two multiplicatively independent
    integers.  Then there exist integers $k,\ell \ge 1$ such that
    $nq^{c+d\ell} \le mp^{a+bk} < (m+1)p^{a+bk} \le (n+1)q^{c+d\ell}$.
\end{lemma}

\begin{proof}
  It is enough to find integers $k,\ell$ satisfying
  $$\frac{nq^c}{mp^a}\le \frac{(p^b)^k}{(q^d)^\ell}\le
  \frac{(n+1)q^c}{(m+1)p^a}.$$
  This is a direct consequence of
  Kronecker's theorem (because $p^b$ and $q^d$ are still
  multiplicatively independent hence $\log p^b/\log q^d$ is
  irrational) \cite{HW}.
\end{proof}

\begin{lemma}\label{lem:3}
    Let $p\ge 2$ and $X\subseteq\mathbb{N}$ be an infinite
    $p$-recognizable set. Then there exist integers $m,a,b\ge 1$ such
    that for all $k\in\mathbb{N}$, the set
    $X\cap[mp^{a+bk},(m+1)p^{a+bk}[$ is nonempty. Moreover, the
    integer $m$ can be chosen arbitrarily large.
\end{lemma}

\begin{proof}
    Let $\mathcal{A}=(Q,q_0,F,\Sigma,\delta)$ be a DFA recognizing
    $\rho_p(X)$. Since $X$ is infinite, there exists $m>0$ arbitrarily
    large such that $\rho_p(m)$ is prefix of an infinite number of
    elements in $\rho_p(X)$. Let $s=\delta(q_0,\rho_p(m))$. By Lemma
    \ref{lem:1}, there exist $\alpha\ge 0$ and $b>0$ such that
    $\mathbf{1}_{L_s}(n)=\mathbf{1}_{L_s}(n+b)$ for all $n\ge \alpha$.
  
  For any $t\ge 0$, the interval $[mp^{t},(m+1)p^t[$ contains all the
  integers having a $p$-ary representation of the form $\rho_p(m)w$
  with $|w|=t$. Since the set $(\rho_p(m)\Sigma_p^*)\cap \rho_p(X)$ is
  infinite, there exists a word $v$ such that $\rho_p(m)v$ is the
  $p$-ary representation of an element in $X$ with $|v|> \alpha$.
  Take $a=|v|$.  Consequently, the interval $[mp^a,(m+1)p^a[$ contains
  an element belonging to $X$. The conclusion follows from the
  periodicity of $\mathbf{1}_{L_s}$:
  $\mathbf{1}_{L_s}(a)=\mathbf{1}_{L_s}(a+kb)=1$, for all $k\ge 0$.
\end{proof}

Recall that a state $s$ is said to be {\it accessible} (resp. {\it
  coaccessible}) if there exists a word $w$ such that
$\delta(q_0,w)=s$ (resp. $\delta(s,w)\in F$). The {\it trimmed}
minimal automaton of a language $L$ is obtained by taking only states
which are accessible and coaccessible.

\begin{lemma}\label{lem:4}
  Let $p\ge 2$ and $X\subseteq\mathbb{N}$ be an infinite
  $p$-recognizable set such that
  $\mathcal{A}=(Q,q_0,F,\Sigma_p,\delta)$ is the trimmed minimal
  automaton of $\rho_p(X)$. If there exists a state $s$ such that
  $\mathbb{N}\setminus L_s$ is infinite, then there exist integers
  $m,a,b\ge 1$ such that for all $k\in\mathbb{N}$, the set
  $X\cap[mp^{a+bk},(m+1)p^{a+bk}[$ is empty.
\end{lemma}

\begin{proof}
  Let $s$ be a state such that $\mathbb{N}\setminus L_s$ is infinite.
  Without loss of generality, we may assume that $s\neq q_0$ and there
  exists $m>0$ such that $\delta(q_0,\rho_p(m))=s$. (Indeed, if
  $\mathbb{N}\setminus L_{q_0}$ is infinite then the same property
  holds for some other state $s$.) We use the same reasoning as in the
  previous proof. Thanks to Lemma~\ref{lem:1}, there exist $\alpha\ge
  0$ and $b>0$ such that $\mathbf{1}_{L_s}(n)=\mathbf{1}_{L_s}(n+b)$
  for all $n\ge \alpha$.  Since $\mathbb{N}\setminus L_s$ is infinite,
  there exists $a> \alpha$ such that no word $v$ of length $a$ is
  such that $\delta(s,v)\in F$. In other words, if $|v|=a$ then
  $\rho_p(m)v\not\in\rho_p(X)$ and the interval $[mp^a,(m+1)p^a[$ does
  not contain any element of $X$. Once again, the conclusion follows
  from the periodicity of $\mathbf{1}_{L_s}$.
\end{proof}

The last lemma is a simple consequence of the three previous ones.

\begin{lemma}\label{lem:5}
  Let $q>p\ge 2$ be two multiplicatively independent integers and
  $X\subseteq\mathbb{N}$ be an infinite $p$- and $q$-recognizable set
  of integers. If $\mathcal{A}=(Q,q_0,F,\Sigma_p,\delta)$ is trimmed
  minimal automaton of $\rho_q(X)$, then for any state $r\in Q$, the
  set $L_r$ is cofinite.
\end{lemma}

\begin{proof}
    Assume to the contrary that $\mathbb{N}\setminus L_r$ is infinite.
    By Lemma~\ref{lem:4}, there exist $n,c,d\ge 1$ such that for
    all $\ell\in\mathbb{N}$, $X\cap[nq^{c+d\ell},(n+1)q^{c+d\ell}[$ is
    empty.
    
    By Lemma~\ref{lem:3}, there also exist $m,a,b\ge 1$ such that for
    all $k\in\mathbb{N}$, $X\cap[mp^{a+bk},(m+1)p^{a+bk}[$ is nonempty
    and $m>n$.

  To obtain a contradiction, simply observe that as a consequence of
  Lemma~\ref{lem:2}, there exist $K,L \ge 1$ such that
  $nq^{c+dL} \le mp^{a+bK} < (m+1)p^{a+bK} \le (n+1)q^{c+dL}$.
\end{proof}

We now have at our disposal all the necessary material to conclude
this short note.
\begin{proof}[Proof of Proposition \ref{pro:1}]
  Assume that $q>p$. Let $\mathcal{A}=(Q,q_0,F,\Sigma,\delta)$ be the
  trimmed minimal automaton of $\rho_q(X)$. For all $n> 0$, we write
  $q_n:=\delta(q_0,\rho_q(n))$.  Thanks to Lemma~\ref{lem:5},
  $L_{q_n}$ is cofinite. This means that for all $n\ge 0$, there
  exists $C_n$ such that for all $k\ge C_n$, $k$ belongs to $L_{q_n}$.
  Clearly, $C_n$ depends only on the state $q_n$ and there are a
  finite number of such states. Let $C=\max\{C_n\}$. Consequently, for
  any $n>0$, there exists a word $w_n$ of length $C$ such that
  $\rho_q(n)w_n\in \rho_q(X)$. In other words, for any $n>0$, there
  exist $t_n\in[0,q^C[$ such that $nq^C+t_n\in X$. We conclude that
  any interval of length $2q^C$ contains at least an element belonging
  to $X$.
\end{proof}

\end{document}